\documentclass[12pt]{article}

\usepackage{amsmath}
\usepackage{amssymb}
\usepackage{amsthm}
\usepackage{hyperref}
\usepackage{enumerate}

\newtheorem{thm}{Theorem}[section]

\newtheorem{lem}[thm]{Lemma}

\newtheorem{assumption}[thm]{Assumption}

\newtheorem{definition}[thm]{Definition}

\newtheorem{example}[thm]{Example}
\newenvironment{exmp}{\begin{example}\rm}{\end{example}}
\newtheorem{remark}[thm]{Remark}
\newenvironment{rem}{\begin{remark}\rm}{\end{remark}}

\newcommand{\E}{\mathbb{E}}
\newcommand{\eps}{\varepsilon}

\title{Capacity Regions of Families of Continuous-Time Multi-User Gaussian Channels}

\author{\small \begin{tabular}{ccc}
Xianming Liu & Guangyue Han\\
The University of Hong Kong & The University of Hong Kong\\
email: mathxmliu@gmail.com & email: ghan@hku.hk\\
\end{tabular}}

\date{{\normalsize \today}}

\begin{document} \maketitle \pagestyle{plain}

\begin{abstract}
In this paper, we propose to use Brownian motions to model families of continuous-time multiuser Gaussian channels without bandwidth limit. It turns out that such a formulation allows parallel translation of many fundamental notions and techniques from the discrete-time setting to the continuous-time regime, which enables us to derive the capacity regions of a continuous-time white Gaussian multiple access channel with/without feedback, a continuous-time white Gaussian interference channel without feedback and a continuous-time white Gaussian broadcast channel without feedback. In theory, these capacity results give the fundamental transmission limit modulation/coding schemes can achieve for families of continuous-time Gaussian one-hop channels without bandwidth limit; in practice, the explicit capacity regions derived and capacity achieving modulation/coding scheme proposed may provide engineering insights on designing multi-user communication systems operating on an ultra-wideband regime.
\end{abstract}


\section{Introduction}  \label{intro}

Continuous-time channels were considered at the very inception of information theory. In his celebrated paper~\cite{sh48} birthing information theory, Shannon studied the following point-to-point continuous-time white Gaussian channels:
\begin{equation} \label{white-Gaussian-noise-channel}
Y(t)=X(t)+Z(t), \quad t \in \mathbb{R},
\end{equation}
where $X(t)$ is the channel input with average power limit $P$, $Z(t)$ is the white Gaussian noise and $Y(t)$ are the channel output. Shannon actually only considered the case that the channel has bandwidth limit $W$, namely, the channel input $X$ and the noise $Z$, and therefore the output $Y$ all have bandwidth limit $W$ (alternatively, as in ($9.54$) of~\cite{co2006}, this can be interpreted as the original channel (\ref{white-Gaussian-noise-channel}) concatenated with an ideal bandpass filter with bandwidth limit $W$).

Then, assuming $Z$ has flat power spectral density $1$, employing the celebrated sampling theorem~\cite{ny24, sh49}, the continuous-time channel (\ref{white-Gaussian-noise-channel}) can be equivalently represented by a parallel Gaussian channel:
\begin{equation} \label{sampled-white-Gaussian-noise-channel}
Y_{n}^{(W)}=X_{n}^{(W)}+Z_{n}^{(W)}, \quad n \in \mathbb{Z},
\end{equation}
where the noise process $\{Z_n^{(W)}\}$ is i.i.d. with variance $1$. Regarding the ``space'' index $n$ as time, the above parallel channel can be interpreted as a discrete-time Gaussian channel associated with the continuous-time channel (\ref{white-Gaussian-noise-channel}). It is well-known from the theory of discrete-time Gaussian channels that the capacity (per second) of the channel (\ref{sampled-white-Gaussian-noise-channel}) can be computed as
\begin{equation} \label{finite-bandwidth-capacity}
C^{(W)}=W \log \left(1+\frac{P}{2W}\right).	
\end{equation}
Intuitively, the {\it infinite bandwidth capacity} $C$ of the channel (\ref{white-Gaussian-noise-channel}), i.e., the channel capacity without bandwidth limit, can then be computed by taking the limit of the above expression as $W$ tends to infinity:
\begin{equation} \label{infinite-bandwidth-capacity}
C=\lim_{W \to \infty} C^{(W)}=P/2.
\end{equation}

Moments of reflection, however, reveals that the sampling approach as above for the channel capacity (with bandwidth limit or not) is heuristic in nature: For one thing, a bandwidth limited signal cannot be time limited, which renders it infeasible to define the data transmission rate if assuming a channel has bandwidth limit. In this regard, rigorous treatments coping with this issue and other technicalities can be found in~\cite{wy66, ga68}; see also~\cite{sl76} for a relevant in-depth discussion. Another issue is that, even disregarding the above nuisance arising from the bandwidth assumption, the sampling approach only gives a lower bound for the infinite bandwidth capacity of (\ref{white-Gaussian-noise-channel}): it shows that $P/2$ is achievable via a class of special coding schemes, but it is not clear that why transmission rate higher than $P/2$ cannot be achieved by other coding schemes. The infinite bandwidth capacity of (\ref{white-Gaussian-noise-channel}) was rigorously studied in~\cite{fo61, be62}, and a complete proof establishing $P/2$ as its de facto infinite bandwidth capacity can be found in~\cite{as63,as64}.

Extending Shannon's fundamental theorems on point-to-point communication channels to general networks with multiple sources and destinations, network information theory aims to establish the fundamental limits on information flows in networks and the optimal coding schemes that achieve these limits. The vast majority of researches on network information theory to date have been focusing on networks in discrete time. In a way, this phenomenon can find its source from Shannon's original treatment of continuous-time point-to-point channels, where such channels were examined through their associated discrete-time versions. This insightful viewpoint has exerted major influences on the bulk of the related literature on continuous-time Gaussian channels, oftentimes prompting a model shift from the continuous-time setting to the discrete-time one right from the beginning of a research attempt. 

Not surprisingly, the same sampling approach can be applied to continuous-time white Gaussian channels with multiple senders/receivers. The following three examples derive, in a heuristic way, the capacity regions of some multi-user one-hop Gaussian channels.

\begin{exmp} (see also Exercise $15.26$ in~\cite{co2006})
Consider the following continuous-time white Gaussian multiple access channel (MAC) with two senders:
\begin{equation} \label{example-mac}
Y(t)=X_1(t)+X_2(t)+Z(t), \quad t \in \mathbb{R},
\end{equation}
where $X_i$, $i=1, 2$, is the input from the $i$-th user with average power limit $P_i$. Similarly as before, consider its associated discrete-time version corresponding to bandwidth limit $W$:
$$
Y_n=X_{1, n}^{(W)}+X_{2, n}^{(W)}+Z_n^{(W)}, \quad n \in \mathbb{Z}.
$$
Then, it is well known~\cite{el11} that the outer bound on the capacity region can be computed as
$$
\left\{(R_1, R_2) \in \mathbb{R}_+^2: R_1 \leq W \log \left(1+\frac{P_1}{2 W} \right), R_2 \leq W \log \left(1+\frac{P_2}{2 W} \right)\right\},
$$
and the inner bound as
$$
\hspace{-0.6cm} \left\{(R_1, R_2) \in \mathbb{R}_+^2: R_1 \leq W \log \left(1+\frac{P_1}{2 W} \right), R_2 \leq W \log \left(1+\frac{P_2}{2 W} \right), R_1+R_2 \leq W \log \left(1+\frac{P_1+P_2}{2 W} \right)\right\}.
$$
(Here, it is known~\cite{wy74, co75} that the outer bound can be tightened to coincide with the inner bound, which, however, is not needed for this example.) It is easy to verify that the two bounds also collapse into the same region as $W$ tends to infinity:
$$
\left\{(R_1, R_2) \in \mathbb{R}_+^2: R_1 \leq P_1/2, R_2 \leq P_2/2\right\},
$$
which is expected to be the infinite bandwidth capacity region of (\ref{example-mac}). And a similar argument holds for more than two senders as well through a parallel extension.
\end{exmp}

\begin{exmp}
Consider the following white Gaussian interference channel (IC) with two senders and two receivers:
\begin{equation} \label{example-ic}
Y_1(t)=a_{1, 1} X_1(t)+a_{1, 2} X_2(t)+Z_1(t), \quad Y_2(t)=a_{2, 1} X_1(t)+a_{2, 2} X_2(t)+Z_2(t), \quad t \geq 0,
\end{equation}
where $X_i$, $i=1, 2$, is the input from the $i$-th user with average power limit $P_i$, and  $a_{ij} \in \mathbb{R}$, $i, j=1, 2$, is the channel gain from sender $j$ to receiver $i$. Similarly as before, one can consider its associated discrete-time version with bandwidth limit $W$: 
$$
Y_{1, n}=a_{1, 1} X_{1, n}^{(W)}+a_{1, 2} X_{2, n}^{(W)}+Z_{1, n}^{(W)}, \quad Y_{2, n}=a_{2, 1} X_{1, n}^{(W)}+a_{2, 2} X_{2, n}^{(W)}+Z_{2, n}^{(W)}, \quad n \in \mathbb{Z}.
$$
Then, a routine argument will yield the following outer bound:
$$
\left\{(R_1, R_2) \in \mathbb{R}_+^2: R_1 \leq W \log \left(1+\frac{a_{11}^2 P_1}{2 W} \right), R_2 \leq W \log \left(1+\frac{a_{22}^2 P_2}{2 W} \right)\right\}.
$$
And for such a channel, using the coding scheme where interferences are treated as noise, a lower bound on the capacity region can be derived~\cite{sa78, ha81} as:
$$
\left\{(R_1, R_2) \in \mathbb{R}^2: R_1 \leq W \log \left(1+\frac{a_{11}^2 P_1}{2 W+a_{12}^2 P_2} \right), R_2 \leq W \log \left(1+\frac{a_{22}^2 P_2}{2 W+ a_{21}^2 P_1}\right) \right\}.
$$
Now, letting $W$ tend to infinity, the above outer and inner bounds will collapse into the same region
$$
\left\{(R_1, R_2) \in \mathbb{R}_+^2: R_1 \leq a_{11}^2 P_1/2, R_2 \leq a_{22}^2 P_2/2\right\},
$$
which is expected to be the infinite bandwidth capacity region of (\ref{example-ic}). Moreover, it is clear that the above argument can be naturally generalized to more than two senders and two receivers.
\end{exmp}

\begin{exmp}
Consider the following white Gaussian broadcast channels (BC) with two receivers:
\begin{equation} \label{example-bc}
Y_1(t)=\sqrt{snr_1} X(t)+Z_1(t), \quad Y_2(t)=\sqrt{snr_2} X(t)+Z_2(t), \quad t \in \mathbb{R},
\end{equation}
where, $snr_i$ is the signal-to-noise ratio (SNR) in the channel for user $i$, $i=1, 2$,  and without loss of generality, we have assumed $snr_1 \geq snr_2$.
Assume the input $X$ has average power limit $P$ and the channel has bandwidth limit $W$. Then, the associated discrete-time channel can be characterized by
$$
Y_{1, n}^{(W)}=\sqrt{snr_1} X_{1, n}^{(W)}+Z_{1, n}^{(W)}, \quad Y_{2, n}^{(W)}=\sqrt{snr_2} X_{2, n}^{(W)}+Z_{2, n}^{(W)}, \quad n \in \mathbb{Z}.
$$
A routine argument shows that the capacity region (per second) of the above channel has an outer bound
$$
\left\{(R_1, R_2) \in \mathbb{R}_+^2: R_1 \leq W \log \left(1+\frac{\alpha snr_1 P}{2 W} \right), R_2 \leq W \log \left(1+\frac{(1-\alpha) snr_2 P}{\alpha snr_1 P+ 2 W} \right), \alpha \in [0, 1]\right\}.
$$
On the other hand, by the superposition coding scheme~\cite{co72, be73}, any rate pair $(R_1, R_2)$ as specified above is also achievable. Letting $W$ tend to infinity, the limiting region suggests that the infinite bandwidth capacity region of (\ref{example-bc}) should be
$$
\left\{(R_1, R_2) \in \mathbb{R}_+^2: \frac{R_1}{snr_1}+\frac{R_2}{snr_2} \leq \frac{P}{2} \right\},
$$
Moreover, it is clear that the above argument can be naturally generalized to more than two receivers.
\end{exmp}

The primary focus of this paper is the infinite bandwidth capacity regions of families of continuous-time multi-user one-hop white Gaussian channels with possible feedback, including continuous-time multi-user white Gaussian MACs, ICs and BCs. Though it is very plausible that approaches in the above three examples can be made rigorous through transplanting the ideas and techniques in~\cite{wy66, as63, as64} to the multi-user scenarios, we will directly work within the continuous-time setting, rather than recasting the problem in discrete time. More specifically, we will use Brownian motions to formulate the problems and we will give complete characterizations of the capacity regions of a number of continuous-time one-hop channels with possible feedback under the average power constraints. One of our goals in this paper is to advocate this alternative formulation, which equips us with established tools and techniques from stochastic calculus and enables us to translate some classical ideas and techniques from the discrete-time setting to the continuous-time one.

For instance, to formulate the white Gaussian channel (\ref{white-Gaussian-noise-channel}), instead of using a white Gaussian noise, we would rather follow~\cite{ih93} to use a Brownian motion and consider the following integral version of (\ref{white-Gaussian-noise-channel}):
\begin{equation} \label{Brownian-motion-channel}
Y(t)=\int_0^t X(s)ds + B(t),
\end{equation}
where, slightly abusing the notation, we still use $Y(t)$ to denote the output corresponding to the input $X(s)$, and $B(t)$ denotes the standard Brownian motion ($Z(t)$ can be viewed as a generalized derivative of $B(t)$); equivalently, the channel (\ref{Brownian-motion-channel}) can be seen as the original channel (\ref{white-Gaussian-noise-channel}) concatenated with an integrator circuit. As opposed to white Gaussian noises, which only exist as generalized functions~\cite{po94}, Brownian motions are well-defined stochastic processes and have been extensively studied in probability theory. Here we remark that, via a routine orthonormal decomposition argument, both of the two channels are equivalent to a parallel channel consisting of infinitely many Gaussian sub-channels~\cite{as65}.

An immediate and convenient consequence of such a formulation is that many notions in discrete time, including mutual information and typical sets, carry over to the continuous-time setting, which will rid us of the nuisances arising from the bandwidth limit assumption. Indeed, such a framework yields a clean and direct proof~\cite{ka71} that the capacity of (\ref{Brownian-motion-channel}) is $P/2$; moreover, as evidenced by numerous results collected in~\cite{ih93} on point-to-point Gaussian channels, the use of Brownian motions elevate the level of rigorousness of our treatment, and equip us with a wide range of established techniques and tools from stochastic calculus.

Secondly, as elaborated in Section~\ref{time-sampling}, the Brownian motion formulation will establish a natural connection between continuous-time Gaussian feedback channels and their discrete-time counterparts (cf. the relevant discussion on the origin of discrete-time Gaussian feedback channel in~\cite{yhk06}). Note that such a connection is established through time sampling, which takes advantage of the continuity of sample paths of a Brownian motion and naturally inherits time causality. On the other hand, the white Gaussian noise formulation is facing inherent difficulty as far as inheriting time causality is concerned: in converting (\ref{white-Gaussian-noise-channel}) to (\ref{sampled-white-Gaussian-noise-channel}), while $X_n^{(W)}$ are obtained as ``time'' samples of $X(t)$, $Z_n^{(W)}$ are in fact ``space'' samples of $Z(t)$, as they are merely the coefficients of the Karhunen-Loeve decomposition of $Z(t)$; see~\cite{yhk06} for an in-depth discussion on this.

Below, we summarize the results obtained in this paper. To put our results into a relevant context, we will first list some relevant results in discrete time or in the (ultra)-wideband regime.

{\bf Gaussian MACs.} When there is no feedback, the capacity region of a discrete-time memoryless MAC is relatively better understood: a single-letter characterization has been established by Ahlswede~\cite{ah73} and the capacity region of a Gaussian MAC was explicitly derived in Wyner~\cite{wy74} and Cover~\cite{co75}. On the other hand, the capacity region of MACs with feedback still demands more complete understanding, despite several decades of great effort by many authors: Cover and Leung~\cite{co81} derived an achievable region for a memoryless MAC with feedback. In~\cite{wi82}, Willems showed that Cover and Leung's region is optimal for a class of memoryless MACs with feedback where one of the inputs is a deterministic function of the output and the other input. More recently, Bross and Lapidoth~\cite{br05} improved Cover and Leung's region, and Wu {\em et al.}~\cite{wu} extended Cover and Leung's region for the case where non-causal state information is available at both senders. An interesting result has been obtained by Ozarow~\cite{oz84}, who derived the capacity region of a memoryless Gaussian MAC with two users via a modification of the Schalkwijk-Kailath scheme~\cite{sc66}. Moreover, Ozarow's result showed that in general, the capacity region for a discrete memoryless MAC is increased by feedback. The capacity region of more general MACs has also been considered; see, e.g.,~\cite{la06, ma, me06, ve93, ch93, kr98, kr03, ha03, ha98, pe09} and references therein. Unfortunately, none of the above-mentioned work gives an explicit characterization of the capacity region of a generic multiple access channel with feedback, which is widely believed to be highly intractable.

In Section~\ref{MAC}, we derive the infinite bandwidth capacity region of a continuous-time white Gaussian MAC with $m$ senders and with/without feedback. It turns out that for such a channel, the feedback does not increase the capacity region. \hfill $\blacksquare$

{\bf Gaussian ICs.} The capacity regions of discrete-time Gaussian ICs are largely unknown except for certain special scenarios: The capacity region of Gaussian ICs with strong interference has been established in Sato~\cite{sa78}, Han and Kobayashi~\cite{ha81}. The sum-capacity of Gaussian ICs with weak interference has been simultaneously derived in~\cite{sh09, an09, mo09}. The half-bit theorem on the tightness of the Han-Kobayashi bound~\cite{ha81} was proven in~\cite{et08}. The approximation of the Gaussian IC by the $q$-ary expansion deterministic channel was first proposed by Avestimehr, Diggavi, and Tse~\cite{av11}. Note that all the above-mentioned work deal with ICs with two pairs of senders and receivers. For more than two user pairs, special classes of Gaussian ICs have been examined using the scheme of interference alignment; see an extensive list of references in~\cite{el11}.

In Section~\ref{IC}, we derive the infinite bandwidth capacity region of a continuous-time white Gaussian IC with $m$ pairs of senders and receivers and without feedback. \hfill $\blacksquare$

{\bf Gaussian BCs.} The capacity regions of discrete-time Gaussian BCs without feedback are well known~\cite{co72, be73}. And it has been shown by El Gamal~\cite{el81} that feedback cannot increase the capacity region of a physically degraded Gaussian BC. On the other hand, it was shown by Ozarow and Leung~\cite{ozle84} that feedback can increase the capacity of stochastically degraded Gaussian BCs, whose capacity regions are far less understood.

In Section~\ref{BC}, we derive the infinite bandwidth capacity region of a continuous-time BC with $m$ receivers and without feedback. \hfill $\blacksquare$

{\bf Wideband Multi-User Gaussian Channels.} The literature on multi-user Gaussian networks operating at the (ultra-)wideband regime is vast, and much work has been focused on the interplay between bandwidth, power and transmission rate. Of greater relevance to this work are~\cite{mc87, ve90, la03, ve02a}, where, among many other results, the asymptotic behavior of multi-user Gaussian broadcast channels, as the bandwidth tends to infinity (or equivalently, as the power tends to zero), is discussed in great depth. 

Note that our results are coding theorems when the channel bandwidth ``is at infinity'', rather than ``tends to infinity'', which is one of the key differences between the above-mentioned work and ours (the difference is subtle and even somewhat confusing as the same notion ``infinite bandwidth'' was at times interpreted differently). In practice, the explicit capacity regions derived and capacity achieving modulation/coding scheme proposed may provide engineering insights on designing multiuser communication systems operating on an (ultra-)wideband regime.

Here we remark that for continuous-time channels, there are alternative definitions of the capacity regions. Counterpart results for all the aforementioned channels under these alternative definitions of the capacity regions are given in Section~\ref{Repeated-Channel}.

\section{Sampling Theorems for Infinite Bandwidth Gaussian Channels} \label{time-sampling}

In this section, we will prove sampling theorems for infinite bandwidth continuous-time white Gaussian feedback channels with possible feedback, which naturally connects such channels with their discrete-time versions.

We first establish a sampling theorem for the following continuous-time white Gaussian feedback channel:
\begin{equation} \label{before-sampling-with-feedback}
Y(t)=\int_0^t X(s, M, Y_0^{s}) ds + B(t), \quad t \in [0, T],
\end{equation}
where $X$ is the channel input, which depends on $M$, the message uniformly distributed over a finite alphabet $\mathcal{M}$, and $Y_0^{s}$, the channel output up to time $s$, and is continuous in $s$ for any fixed realizations of $M$ and $Y_0^s$. Note that, strictly speaking, the third parameter of $X$ in (\ref{before-sampling-with-feedback}) should be $Y_0^{s-}$, which, however, can be equivalently replaced by $Y_0^s$ due to the continuity of sample paths of $\{Y(t)\}$. 

Assume that the following regularity conditions are satisfied:
\begin{itemize}
\item[(a)] The solution $\{Y(t)\}$ to the stochastic differential equation (\ref{before-sampling-with-feedback}) uniquely exists;
\item[(b)] Novikov's condition~\cite{ok95} on $X(s, M, Y_0^s)$, that is,
\begin{equation}  \label{Novikov-Condition}
\E\left[\exp\left(\frac{1}{2} \int_0^T X^2(s, M, Y_0^s) ds\right)\right] < \infty.
\end{equation}
\end{itemize}

Now, for any $n \in \mathbb{N}$, choose $t_{n, 0}, t_{n, 1}, \ldots, t_{n, n}$ such that
$$
0=t_{n, 0} < t_{n, 1} < \cdots < t_{n, n-1} < t_{n, n}=T,
$$
and let $\Delta_n \triangleq \{t_{n, 0}, t_{n, 1}, \ldots, t_{n, n}\}$. Sampling the channel over the time interval $[0, T]$ with respect to $\Delta_n$, we then obtain a {\em sampled version} of the continuous-time channel (\ref{before-sampling-with-feedback}):
\begin{equation}  \label{after-sampling-with-feedback}
Y(t_{n, i})=\int_0^{t_{n, i}} X(s, M, Y_0^{s-}) ds + B(t_{n, i}), \quad i=0, 1, \ldots, n.
\end{equation} 
The sequence $\{\Delta_n\}$ is said to be {\em increasingly refined} if $\Delta_n \subset \Delta_{n+1}$ for any $n \in \mathbb{N}$. Roughly speaking, the following sampling theorem states that for any sequence of increasingly refined samplings, the mutual information of the sampled channel (\ref{after-sampling-with-feedback}) will converge to that of the original channel (\ref{before-sampling-with-feedback}).
\begin{thm} \label{sampling-output}
For any increasingly refined $\{\Delta_n\}$, we have
$$
\lim_{n \to \infty} I(M; Y_{\Delta_n})=I(M; Y_0^T),
$$
where $Y_{\Delta_n} \triangleq \{Y_{t_{n, 0}}, Y_{t_{n, 1}}, \ldots, Y_{t_{n, n}}\}$.
\end{thm}

\begin{proof}
First of all, it follows from Theorem $7.1$ of~\cite{li01} (it can be checked that its assumptions are implied by Condition (b)) that for any $m \in \mathcal{M}$, $\mu_Y \sim \mu_{Y|M=m} \sim \mu_B$, where ``$\sim$'' means ``equivalent'', and moreover,
$$
\frac{d\mu_{Y|M}}{d\mu_B}(Y_0^T)=\frac{1}{\E[e^{-\int_0^T X dY+1/2 \int_0^T X^2 ds}|Y_0^T, M]}, \quad \frac{d\mu_{Y}}{d\mu_B}(Y_0^T)=\frac{1}{\E[e^{-\int_0^T X dY+1/2 \int_0^T X^2 ds}|Y_0^T]}.
$$
Here we remark that $\E[e^{-\int_0^T X dY+1/2 \int_0^T X^2 ds}|Y_0^T, M]$ is in fact equal to $e^{-\int_0^T X dY+1/2 \int_0^T X^2 ds}$, but we keep it the way it is as above for later comparison. A parallel argument as in the proof of Theorem $7.1$ further implies that, for any $\Delta_n$,
$$
\frac{d\mu_{Y|M}}{d\mu_B}(Y_{\Delta_n})=\frac{1}{\E[e^{-\int_0^T X dY+1/2 \int_0^T X^2 ds}|Y_{\Delta_n}, M]}, \quad \frac{d\mu_{Y}}{d\mu_B}(Y_{\Delta_n})=\frac{1}{\E[e^{-\int_0^T X dY+1/2 \int_0^T X^2 ds}|Y_{\Delta_n}]},
$$
where, similarly as before, we have defined (recall that $Y_{\Delta_n} \triangleq \{Y_{t_{n, 0}}, \ldots, Y_{t_{n, n}}\}$)
$$
\quad B_{\Delta_n} \triangleq \{B_{t_{n, 0}}, \ldots, B_{t_{n, n}}\}, $$ 
and moreover,
$$
\frac{d\mu_{Y|M}}{d\mu_B}(Y_{\Delta_n}) \triangleq \frac{d\mu_{Y_{\Delta_n}|M}}{d\mu_{B_{\Delta_n}}}(Y_{\Delta_n}), \quad \frac{d\mu_{Y}}{d\mu_B}(Y_{\Delta_n}) \triangleq
\frac{d\mu_{Y_{\Delta_n}}}{d\mu_{B_{\Delta_n}}}(Y_{\Delta_n}).
$$
Then, by definition, we have
$$
I(M; Y_{\Delta_n}) = \E\left[\log \frac{d\mu_{Y|M}}{d\mu_B}(Y_{\Delta_n})\right]-\E\left[\log \frac{d\mu_{Y}}{d\mu_B}(Y_{\Delta_n})\right].
$$
Notice that Novikov's condition implies that $e^{-\int_0^T X dY+1/2 \int_0^T X^2 ds}$ integrable, which further implies that $\{\frac{d\mu_{Y|M}}{d\mu_B}(Y_{\Delta_n})\}$ and $\{\frac{d\mu_{Y}}{d\mu_B}(Y_{\Delta_n})\}$ are both martingales, and therefore,
$$
\frac{d\mu_{Y|M}}{d\mu_B}(Y_{\Delta_n}) \to \frac{d\mu_{Y|M}}{d\mu_B}(Y_0^T), \quad \frac{d\mu_{Y}}{d\mu_B}(Y_{\Delta_n}) \to \frac{d\mu_{Y}}{d\mu_B}(Y_0^T), \mbox{ a.s.}
$$ 

Now, by Jensen's inequality, we have
$$
\E\left[ \left. -\int^T_0 X dY_s+\frac{1}{2}\int^T_0
X^2ds \right| Y_{\Delta_n}, M \right] \leq \log\E[e^{-\int^T_0 X
dY_s+\frac{1}{2}\int^T_0 X^2ds}|Y_{\Delta_n}, M],
$$
and, by the easy fact that $\log x \leq x$ for any $x > 0$,
$$
\log\E[e^{-\int^T_0 X dY_s+\frac{1}{2}\int^T_0 X^2ds}|Y_{\Delta_n}, M] \leq
\E[e^{-\int^T_0 X dY_s+\frac{1}{2}\int^T_0
X^2ds}|Y_{\Delta_n}, M],
$$
It then follows that
$$
\hspace{-1.7cm} \left|\log\E[e^{-\int^T_0 X dY_s+\frac{1}{2}\int^T_0 X^2ds}|Y_{\Delta_n}, M]\right| \leq
\left|\E\left[\left. -\int^T_0 X dY_s+\frac{1}{2}\int^T_0 X^2 ds \right|Y_{\Delta_n}, M\right]\right|+\E[e^{-\int^T_0 X dY_s+\frac{1}{2}\int^T_0 X^2ds}|Y_{\Delta_n}, M].
$$
Applying the generalized dominated convergence theorem (see, e.g., Theorem $19$ on Page $89$ of~\cite{ro10}), we then have
$$
\lim_{n \to \infty} \E\left[\log \frac{d\mu_{Y|M}}{d\mu_B}(Y_{\Delta_n})\right] = \E [\log  \E[e^{-\int^T_0 X dY_s+\frac{1}{2}\int^T_0 X^2ds}|Y^{T}_{0}, M]] = \E\left[\log \frac{d\mu_{Y|M}}{d\mu_B}(Y_{0}^{T})\right].
$$
A completely parallel argument yields that
$$
\lim_{n \to \infty} \E\left[\log \frac{d\mu_{Y}}{d\mu_B}(Y_{\Delta_n})\right] = \E[\log  \E[e^{-\int^T_0 X dY_s+\frac{1}{2}\int^T_0 X^2ds}|Y^{T}_{0}]] = \E\left[\log \frac{d\mu_{Y}}{d\mu_B}(Y_{0}^{T})\right].
$$
So, with the fact
$$
I(M; Y_0^T)=\E\left[\log \frac{d\mu_{Y|M}}{d\mu_B}(Y_0^T)\right]-\E\left[\log \frac{d\mu_{Y}}{d\mu_B}(Y_0^T)\right],
$$
we conclude that 
$$
\lim_{n \to \infty} I(M; Y_{\Delta_n})=\E\left[\log \frac{d\mu_{Y|M}}{d\mu_B}(Y_0^T)\right]-\E\left[\log \frac{d\mu_{Y}}{d\mu_B}(Y_0^T)\right]=I(M;Y_0^T).
$$
\end{proof}

Next, we will establish a sampling theorem for the following continuous-time white Gaussian channel without feedback:
\begin{equation} \label{before-sampling-without-feedback}
Y(t)=\int_0^t X(s, M) ds + B(t), \quad t \in [0, T],
\end{equation}
where $X$ is the channel input, which depends on $M$, the message uniformly distributed over a finite alphabet $\mathcal{M}$, and is continuous in $s$ for any fixed realization of $M$, and satisfies the following power constraint:
$$
\int_0^T \E^2[X(s, M)] ds < \infty.
$$ 
Letting $\tilde{X}(t)=\int_0^t X(s, M) ds$ and sampling the channel (\ref{before-sampling-without-feedback}) with respect to $\Delta_n$, we have the following sampled channel:
\begin{equation} \label{after-sampling-without-feedback}
Y(t_{n, i})=\tilde{X}(t_{n, i}) + B(t_{n, i}), \quad i=1, 2, \ldots, n.
\end{equation}
It turns out that when the feedback is absent, we have the following sampling theorem, in which, as opposed to Theorem~\ref{sampling-output}, both the input and output are sampled.

\begin{thm} \label{sampling-input-output}
For any increasingly refined $\{\Delta_n\}$, we have
$$
\lim_{n \to \infty} I(\tilde{X}_{\Delta_n}; Y_{\Delta_n})=I(X_0^T; Y_0^T),
$$
where $\tilde{X}_{\Delta_n} \triangleq \{\tilde{X}_{t_{n, 0}}, \tilde{X}_{t_{n, 1}}, \ldots, \tilde{X}_{t_{n, n}}\}$ and $Y_{\Delta_n} \triangleq \{Y_{t_{n, 0}}, Y_{t_{n, 1}}, \ldots, Y_{t_{n, n}}\}$.
\end{thm}

\begin{proof}
We will follow~\cite{HanSong2014} and consider a family of continuous-time channels parameterized by $\rho \geq 0$ and their sampled versions with respect to $\Delta_n$:
$$
Y(t)=\rho \int_0^t X(s, M)ds+B(t),\quad Y(t_{n, i})=\rho \int_0^{t_{n, i}} X(s, M)ds+B(t_{n, i}), \quad i=1, 2, \ldots, n,
$$
where $\rho \geq 0$. Here, we remark that, within this proof only, we have adopted the same notation in (\ref{before-sampling-without-feedback}) for the parameterized channel, which obviously depends on $\rho$. Note that if $\rho=1$, the above parameterized channel boils down to the channel (\ref{before-sampling-without-feedback}).

It suffices  to prove that for the paratermized channel, the derivative of the mutual information of the sampled versions with respect to $\rho$ converges to that of the original channel. To be more precise, we will show that for any $\rho \geq 0$,
$$
\lim_{n \rightarrow \infty} \frac{d}{d \rho}I(\tilde{X}_{\Delta_n}; Y_{\Delta_n}) =\frac{d}{d \rho} I(M; Y^T_0), 
$$
which, together with the well-known fact (see, e.g., Theorem $6.2.1$ of~\cite{ih93}) that
$$
I(M; Y_0^T)=I(X_0^T; Y_0^T)=\frac{\rho^2}{2} \int_0^T \E[X^2(s, M)]-\E[\E^2[X(s, M)|Y_0^s]] ds,
$$
and the dominated convergence theorem, will immediately imply the theorem.

Now, using the fact that
$$
f_{Y|\tilde{X}}(y_{\Delta_n}|\tilde{x}_{\Delta_n})=
\prod_{i=1}^{n} \frac{1}{\sqrt{2\pi(t_{n, i}-t_{n, i-1})}}\exp \left({-\frac{(y_{t_{n, i}}-y_{t_{n, i-1}}- \rho \int_{t_{n, i-1}}^{t_{n, i}} \tilde{x}(s) ds)^2}{2(t_{n, i}-t_{n, i-1})}} \right),
$$
we have
\begin{align*}
\frac{d}{d \rho} f_Y(Y_{\Delta_n}) &= \int \frac{d}{d
\rho}f_{Y|\tilde{X}}(Y_{\Delta_n}|\tilde{x}_{\Delta_n}) f_{\tilde{X}}(\tilde{x}_{\Delta_n}) d\tilde{x}_{\Delta_n}\\
&= \int \sum_{i=1}^n \frac{(Y_{t_{n, i}}-Y_{t_{n, i-1}}-\int_{t_{n, i-1}}^{t_{n, i}}
\rho \tilde{x}(s) ds) \int_{t_{n, i-1}}^{t_{n, i}} \rho \tilde{x}(s) ds
}{t_{n, i}-t_{n, i-1}} f_{Y|\tilde{X}}(Y_{\Delta_n}|\tilde{x}_{\Delta_n}) f_{\tilde{X}}(\tilde{x}_{\Delta_n}) d\tilde{x}_{\Delta_n}.
\end{align*}
It then follows that
\begin{align*}
\frac{d}{d\rho}I(\tilde{X}_{\Delta_n}; Y_{\Delta_n}) &= -\E\left[ \frac{1}{f(Y_{\Delta_n})}\frac{d}{d\rho}f(Y_{\Delta_n}) \right] \\
&= \sum_{i=1}^n \frac{\rho \E[(\int_{t_{n, i-1}}^{t_{n, i}} X(s, M)ds)^2]}{t_{n, i}-t_{n, i-1}}-\sum_{i=1}^n \frac{\rho \E[\E^2[\int_{t_{n, i-1}}^{t_{n, i}}
X(s, M)ds|Y_{\Delta_n}]]}{t_{n, i}-t_{n, i-1}}.
\end{align*}
Using the fact that $X(s, M)$ is continuous in $s$ and the message alphabet $\mathcal{M}$ is finite, we have, 
\begin{align*}
\lim_{n \to \infty} \sum_{i=1}^{n} \frac{\E[(\int_{t_{n, i-1}}^{t_{n, i}}
X(s, M)ds)^2]}{t_{n, i}-t_{n, i-1}} & = \lim_{n \to \infty} \sum_{i=1}^{n} \frac{\E[(\int_{t_{n, i-1}}^{t_{n, i}}
X((t_{n, i-1}+t_{n, i})/2, M)ds)^2]}{t_{n, i}-t_{n, i-1}}\\
&= \lim_{n \to \infty} \sum_{i=1}^{n} \E[
X((t_{n, i-1}+t_{n, i})/2, M)^2](t_{n, i}-t_{n, i-1})\\
&= \int_0^t \E[X^2(s, M)] ds. 
\end{align*}
In a similar fashion, we have
\begin{align*}
\lim_{n \to \infty} \sum_{i=1}^n \frac{\E[\E^2[\int_{t_{n, i-1}}^{t_{n, i}}
X(s, M)ds|Y_{\Delta_n}]]}{t_{n, i}-t_{n, i-1}} & = \lim_{n \to \infty} \sum_{i=1}^n \frac{\E[\E^2[\int_{t_{n, i-1}}^{t_{n, i}}
X((t_{n, i-1}+t_{n, i})/2, M)ds|Y_{\Delta_n}]]}{t_{n, i}-t_{n, i-1}} \\
& = \lim_{n \to \infty} \E[\sum_{i=1}^n \E^2[
X((t_{n, i-1}+t_{n, i})/2, M)(t_{n, i}-t_{n, i-1})|Y_{\Delta_n}]] \\
&=\int_0^t \E[\E^2[X(s, M)|Y_0^T]] ds,
\end{align*}
where we have used the dominated convergence theorem for conditional expectations (see Theorem $5.5.9$ of~\cite{Durrett}). Then, using the proven fact~\cite{gu05, HanSong2014} that 
$$
\frac{d}{d\rho}I(M; Y_0^T)=\rho \int_0^t \E[X^2(s, M)] ds-\rho \int_0^t \E[\E^2[X(s, M)|Y_0^T]] ds,
$$
we arrive at
$$
\lim_{n \to \infty} \frac{d}{d\rho}I(\tilde{X}_{\Delta_n}; Y_{\Delta_n}) = \frac{d}{d\rho}I(M; Y_0^T).
$$
The theorem is then proven.
\end{proof}

\begin{rem}
Taking advantage of the continuity of sample paths of a Brownian motion, Theorem~\ref{sampling-output} and~\ref{sampling-input-output} naturally connect continuous-time Gaussian feedback channels with their discrete-time counterparts. The idea can be roughly explained as follows.

Take, for example, the continuous-time Gaussian feedback channel (\ref{before-sampling-with-feedback}). Apparently, its sampled version with respect to $\Delta_n$ can be rewritten as 
\begin{equation} \label{diff-version}
Y_{t_{n, i}}-Y_{t_{n, i-1}}=\int_{t_{n, i-1}}^{t_{n, i}} X(s, M, Y_0^{s}) ds + B(t_{n, i})-B(t_{n, i-1}), \quad i=1, 2, \ldots, n,
\end{equation}
where $B(t_{n, i})-B(t_{n, i-1})$ are independent Gaussian random variables. 
Now, with certain regularity conditions imposed on the channel input $X$, the mutual information of the above sampled channel should be ``close'' to that of the following discrete-time channel:
\begin{equation} \label{approx-diff-version}
Y_{t_{n, i}}-Y_{t_{n, i-1}} = \int_{t_{n, i-1}}^{t_{n, i}} X(s, M, Y_{t_{n, 1}}^{t_{n, n-1}}) ds + B(t_{n, i})-B(t_{n, i-1}), \quad i=1, 2, \ldots, n, 
\end{equation}
Then, Theorem~\ref{sampling-output} implies that as $n$ tends to infinity, the mutual information of (\ref{approx-diff-version}) will converge to that of (\ref{diff-version}). In other words, via time sampling, Theorem~\ref{sampling-output} can be used to establish continuous-time channels as the ``limits'' of discrete-time channels. It is straightforward to see that the same connections can be established when multiple users are present in the communication systems.
\end{rem}

\section{Gaussian MACs}\label{MAC}

Consider a continuous-time white Gaussian MAC with $m$ users, which can be characterized by
\begin{equation}  \label{Equation-MAC}
Y(t)=\int_0^t X_1(s, M_1, Y_0^s) ds+\int_0^t X_2(s, M_2, Y_0^s) ds + \cdots +\int_0^t X_m(s, M_m, Y_0^s) ds+B(t), \quad t \geq 0,
\end{equation}
where $X_i$ is the continuous channel input from sender $i$, which depends on $M_i$, the message sent from sender $i$, which is independent of all messages from other senders, and possibly on the feedback $Y_0^{s}$, the channel output up to time $s$. Note that, with the presence of feedback, the existence and uniqueness of $Y$ is in fact a tricky mathematical problem, however, we will simply assume that all the inputs $X_i$ are appropriately chosen such that $Y$ uniquely exists.

For $T, R_1, \ldots, R_m, P_1, \ldots, P_m > 0$, a $(T, (e^{T R_1}, \ldots, e^{T R_m}), (P_1, \ldots, P_m))$-code for the MAC (\ref{Equation-MAC}) consists of $m$ sets of integers $\mathcal{M}_i=\{1, 2, \ldots, e^{T R_i}\}$, the {\it message alphabet} for user $i$, $i=1, 2, \ldots, m$, and $m$ {\it encoding functions}, $X_i: \mathcal{M}_i \rightarrow C[0, T]$, which satisfy the following power constraint: for any $i=1, 2, \ldots, m$, 
\begin{equation}  \label{PowerConstraint-MAC}
\frac{1}{T} \int_0^T X^2_i(s, M_i, Y_0^s) ds \leq P_i,
\end{equation}
and a {\it decoding function},
$$
g: C[0, T] \rightarrow \mathcal{M}_1 \times \mathcal{M}_2 \times \cdots \times \mathcal{M}_m.
$$
The average probability of error for the above code is defined as
$$
\hspace{-1cm} P_e^{(T)}=\frac{1}{e^{T (\sum_{i=1}^m R_i)}} \sum_{(M_1, M_2, \ldots, M_m) \in \mathcal{M}_1 \times \mathcal{M}_2 \times \cdots \times \mathcal{M}_m} P\{g(Y_0^T) \neq (M_1, M_2, \ldots, M_m)~|~(M_1, M_2, \ldots, M_m) \mbox{ sent}\}.
$$
A rate tuple $(R_1, R_2, \ldots, R_m)$ is said to be {\bf achievable} for the MAC if there exists a sequence of $(T, (e^{T R_1}, \ldots, e^{T R_m}), (P_1, \ldots, P_m))$-codes with $P_e^{(T)} \rightarrow 0$ as $T \rightarrow \infty$. The {\bf capacity region} of the MAC is the closure of the set of all the achievable $(R_1, R_2, \ldots, R_m)$ rate tuples.

The following theorem gives an explicit characterization of the capacity region.
\begin{thm}  \label{Theorem-MAC}
Whether there is feedback or not, the capacity region of the continuous-time white Gaussian MAC (\ref{Equation-MAC}) is
$$
\{(R_1, R_2, \ldots, R_m) \in \mathbb{R}_+^m: R_i \leq P_i/2, \quad i=1, 2, \ldots, m\}.
$$
\end{thm}

In the following, we will give the proof of Theorem~\ref{Theorem-MAC}. For notational convenience only, we will assume $m=2$, the case with a generic $m$ being completely parallel.

We will need the following lemma.
\begin{lem} \label{independent-OUs}
For any $\eps > 0$, there exist two independent Ornstein-Uhlenbeck (OU) processes $\{X_i(s): s \geq 0\}$, $i=1, 2$, satisfying the following power constraint:
\begin{equation} \label{a-p-c}
\mbox{for $i=1, 2$, there exists $P_i > 0$ such that for all $t > 0$,  } \frac{1}{t} \int_0^t E^2[X_i(s)] ds = P_i,
\end{equation}
such that for all $T$,
\begin{equation} \label{comma-sum}
|I_T(X_1, X_2; Y)/T-(P_1+P_2)/2| \leq \eps,
\end{equation}
and
\begin{equation} \label{conditional}
|I_T(X_1; Y|X_2)/T-P_1/2| \leq \eps, \quad |I_T(X_2; Y|X_1)/T-P_2/2| \leq \eps,
\end{equation}
moreover,
\begin{equation} \label{treat-as-noise}
|I_T(X_1; Y)/T-P_1/2| \leq \eps, \quad |I_T(X_2; Y)/T-P_2/2| \leq \eps,
\end{equation}
where
\begin{equation}  \label{to-be-interpreted-as-channel}
Y(t)=\int_0^t X_1(s) ds + \int_0^t X_2(s) ds + B(t), \quad t \geq 0.
\end{equation}
Here (and often in other parts of the paper) the subscript $T$ means that the conditional mutual information is computed over the time period $[0, T]$.
\end{lem}

\begin{proof}
For $a > 0$, consider the following two independent OU process $X_i(t)$, $i=1, 2$, given by
$$
X_i(t)=\sqrt{2a P_i} \int_{-\infty}^t e^{-a(t-s)} dB_i(s),
$$
where $B_i$, $i=1, 2$, are independent standard Brownian motions. Obviously, for $X_i$ defined as above, (\ref{a-p-c}) is satisfied. A parallel version of the proof of Theorem $6.2.1$ of~\cite{ih93} yields that
$$
I_T(X_1, X_2; Y) = I_T(X_1+X_2; Y)= \frac{1}{2} \int_0^T E[(X_1(t)+X_2(t)-E[X_1(t)+X_2(t)|Y_0^t])^2] dt.
$$
It then follows from Theorem $6.4.1$ in~\cite{ih93} (applied to the OU process $X_1(t)+X_2(t)$) that as $a \rightarrow \infty$,
$$
I(X_1, X_2; Y)/T = I(X_1+X_2; Y)/T \rightarrow (P_1+P_2)/2,
$$
uniformly in $T$, which establishes (\ref{comma-sum}).

For $i=1, 2$, define
$$
\tilde{Y}_i(t)=\int_0^t X_i(s) ds + B(t), \quad t > 0.
$$
As in the proof of Theorem $6.4.1$ in~\cite{ih93}, we deduce that for $i=1, 2$, $I_T(X_i; \tilde{Y}_i)/T$ tend to $P_i/2$ uniformly in $T$. Now, since $X_1$ and $X_2$ are independent, we have for any fixed $T$,
$$
I_T(X_1; Y|X_2)=I_T(X_1; \tilde{Y}_1|X_2)=I_T(X_1; \tilde{Y}_1),
$$
and
$$
I_T(X_2; Y|X_1)=I_T(X_2; \tilde{Y}_2|X_1)=I_T(X_2; \tilde{Y}_2),
$$
which immediately implies (\ref{conditional}).

Now, by the chain rule of mutual information,
$$
I_T(X_1, X_2; Y)=I_T(X_1; Y)+ I_T(X_2; Y|X_1)=I_T(X_2; Y)+ I_T(X_1; Y|X_2),
$$
which, together with (\ref{comma-sum}) and (\ref{conditional}), implies (\ref{treat-as-noise}).

\end{proof}

\begin{rem} 
With $X_i$, $i=1, 2$, regarded as channel inputs, (\ref{to-be-interpreted-as-channel}) can be reinterpreted as a white Gaussian MAC. For $i \neq j$, $I(X_i; Y)$, the reliable transmission rate of  $X_i$ when $X_j$ is not known can be arbitrarily close to $I(X_i; Y|X_j)$, the reliable transmission rate of $X_i$ when $X_j$ is known. In other words, for white Gaussian MACs, knowledge about other user's inputs will not help to achieve faster transmission rate, and therefore, they can be simply treated as noises. An more intuitive explanation of this result is as follows: for the OU-process $X_i$ as specified in the proof, its power spectral density can be computed as
$$
f_i(\lambda)=\frac{2 a P_i}{2 \pi (\lambda^2+a^2)},
$$
which is ``negligible'' compared to that of the white Gaussian noise (which is the constant $1$) as $a$ tends to infinity. Lemma~\ref{independent-OUs} is a key ingredient for deriving the capacity regions of white Gaussian MACS, and, as elaborated later in the paper, those of white Gaussian ICs and BCs as well.
\end{rem}

We also need some result on the information stability of continuous-time Gaussian processes. Let $(U, V)  = \{(U(t), V(t)), t \geq 0\}$ be a continuous Gaussian system (which means $U(t), V(t)$ are pairwise Gaussian stochastic processes). Define
$$
\varphi^{(T)}(u, v)=\frac{d\mu_{UV}^{(T)}}{d\mu_U^{(T)} \times \mu_V^{(T)}}(u, v), \qquad (u, v) \in C[0, T] \times C[0, T],
$$
where $\mu_U^{(T)}$, $\mu_V^{(T)}$ and $\mu_{UV}^{(T)}$ denote the probability distributions of $U_0^T$, $V_0^T$ and their joint distribution, respectively. For any $\varepsilon > 0$, we denote by $\mathcal{T}_{\varepsilon}^{(T)}$ the $\varepsilon$-typical set:
$$
\mathcal{T}^{(T)}_{\varepsilon}=\left\{(u, v) \in C[0, T] \times C[0, T]; \frac{1}{T} |\log \varphi^{(T)}(u, v)-I_T(U, V)| \leq \varepsilon \right\}.
$$
The pair $(U, V)$ is said to be {\it information stable}~\cite{pi64} if for any $\varepsilon > 0$,
$$
\lim_{T \to \infty} \mu^{(T)}_{UV}(\mathcal{T}_{\varepsilon}) = 1.
$$

The following theorem is a rephrased version of Theorem 6.6.2. in~\cite{ih93}.
\begin{lem} \label{information-stability}
The Gaussian system $(U, V)$ is information stable provided that
$$
\lim_{T \rightarrow \infty} \frac{I_T(U; V)}{T^2}=0.
$$
\end{lem}
\noindent Lemma~\ref{information-stability} will be used in the proof of Theorem~\ref{Theorem-MAC} to establish, roughly speaking, that almost all sequences are jointly typical.

We are now ready for the proof of Theorem~\ref{Theorem-MAC}
\begin{proof}[Proof of Theorem~\ref{Theorem-MAC}]

\textbf{The converse part.} In this part, we will show that for any sequence of $(T, (e^{T R_1}, e^{T R_2}), (P_1, P_2))$-codes with $P_e^{(T)} \rightarrow 0$ as $T \rightarrow \infty$, the rate pair $(R_1, R_2)$ will have to satisfy
$$
R_1 \leq P_1/2, \qquad R_2 \leq P_2/2.
$$

Fix $T$ and consider the above-mentioned $(T, (e^{T R_1}, e^{T R_2}), (P_1, P_2))$-code. By the code construction, it is possible to estimate the messages $(M_1, M_2)$ from the channel output $Y_0^T$ with a low probability of error. Hence, the conditional entropy of $(M_1, M_2)$ given $Y_0^T$ must be small; more precisely, by Fano's inequality,
$$
H(M_1,M_2|Y_0^T) \leq T(R_1+R_2)P^{(T)}_e +H(P^{(T)}_e) = T \varepsilon_T,
$$
where $\varepsilon_T \rightarrow 0$ as $T \rightarrow \infty$. Then, we have $$
H(M_1|Y^T) \leq H(M_1, M_2|Y^T) \leq T\varepsilon_T, \quad
H(M_2|Y^T) \leq H(M_1, M_2|Y^T) \leq T\varepsilon_T.
$$
Now, we can bound the rate $R_1$ as follows:
\begin{align*}
T R_1 & = H(M_1)\\
      & = I(M_1; Y_0^T)+H(M_1|Y_0^T)\\
      & \leq I(M_1; Y_0^T) + T \varepsilon_T\\
      & \leq H(M_1)-H(M_1|Y_0^T) + T \varepsilon_T \\
      & \leq H(M_1|M_2)-H(M_1|Y_0^T, M_2)+T \varepsilon_T \\
      & = I(M_1; Y_0^T|M_2) + T \varepsilon_T.
\end{align*}

Applying Theorem $6.2.1$ in~\cite{ih93}, we have
$$
I(M_1; Y_0^T|M_2)=\frac{1}{2} E\left[\int_0^T E[(X_1+X_2-\hat{X}_1-\hat{X}_2)^2|M_2] dt \right]=\frac{1}{2} \int_0^T E[(X_1+X_2-\hat{X}_1-\hat{X}_2)^2] dt,
$$
where $\hat{X}_i(t)=E[X_i(t)|Y_0^T, M_2]$, $i=1, 2$. Noticing that $X_2=\hat{X}_2$, we then have
$$
I(M_1; Y_0^T|M_2)=\frac{1}{2} \int_0^T E[(X_1-\hat{X}_1)^2] dt],
$$
which, together with (\ref{PowerConstraint-MAC}), implies that $R_1 \leq P_1/2$. A completely parallel argument will yield that  $R_2 \leq P_2/2$.

\textbf{The achievability part.} In this part, we will show that as long as $(R_1, R_2)$ satisfying
\begin{equation} \label{strictly-less-than}
0 \leq R_1 < P_1/2, \quad 0 \leq R_2 < P_2/2,
\end{equation}
we can find a sequence of $(T, (e^{T R_1}, e^{T R_2}), (P_1, P_2))$-codes with $P_e^{(T)} \rightarrow 0$ as $T \rightarrow \infty$. The argument consists of several steps as follows.

\emph{Codebook generation}: For a fixed $T > 0$ and $\varepsilon > 0$, assume that $X_1$ and $X_2$ are independent OU processes over $[0, T]$ with respective variances $P_1-\varepsilon$ and $P_2-\varepsilon$, and that $(R_1,R_2)$ satisfying (\ref{strictly-less-than}). Generate $e^{T R_1}$ independent codewords $X_{1, i}$, $i \in \{1, 2, \ldots, e^{T R_1}\}$, of length $T$, according to the distribution of $X_1$. Similarly, generate $e^{T R_2}$ independent codewords $X_{2, j}$, $j \in \{1, 2, \ldots, e^{T R_2}\}$, of length $T$, according to the distribution of $X_2$. These codewords (which may not satisfy the power constraint in (\ref{PowerConstraint-MAC})) form the codebook, which is revealed to the senders and the receiver.

\emph{Encoding:} To send message $i \in \mathcal{M}_1$, sender $1$ sends the codeword $X_{1, i}$. Similarly, to send $j \in \mathcal{M}_2$, sender $2$ sends $X_{2, j}$.

\emph{Decoding:} For any fixed $\varepsilon > 0$, let $\mathcal{T}_{\varepsilon}^{(T)}$ denote the set of {\it jointly typical} $(x_1, x_2, y)$ sequences, which is defined as follows:
$$
\mathcal{T}_{\varepsilon}^{(T)}=\{(x_1,x_2, y) \in C[0, T] \times C[0, T] \times C[0, T]:|\log
\varphi_1(x_1, x_2, y)-I_T(X_1,X_2; Y)|\leq  T \varepsilon,
$$
$$
|\log
\varphi_2(x_1, x_2, y)-I_T(X_1; X_2, Y)|\leq  T \varepsilon, |\log \varphi_3(x_1, x_2, y)-I_T(X_2; X_1,Y)|\leq  T \varepsilon \},
$$
where
$$
\varphi_1(x_1, x_2, y)=\frac{d\mu_{X_1 X_2 Y}}{d\mu_{X_1 X_2} \times \mu_{Y}}(x_1, x_2, y),
$$
$$
\varphi_2(x_1, x_2, y)=\frac{d\mu_{X_1 X_2 Y}}{d\mu_{X_1} \times \mu_{X_2 Y}}(x_1,x_2,y),
$$
$$
\varphi_3(x_1, x_2, y)=\frac{d\mu_{X_1 X_2 Y}}{d\mu_{X_2} \times \mu_{X_1 Y}}(x_1, x_2, y).
$$
Here we remark that it is easy to check that the above Randon-Nykodym derivatives are all well-defined; see, e.g., Theorem $7.7$ of~\cite{li01} for sufficient conditions for their existence. Based on the received output $y \in C[0, T]$, the receiver chooses the pair $(i, j)$ such that
$$
(x_{1, i}, x_{2, j}, y) \in \mathcal{T}_{\varepsilon}^{(T)},
$$
if such a pair $(i, j)$ exists and is unique; otherwise, an error is declared. Moreover, an error will be declared if the chosen codeword does not satisfy the power constraint in (\ref{PowerConstraint-MAC}).

\emph{Analysis of the probability of error:} Now, for fixed $T, \varepsilon > 0$, define
$$
E_{ij}=\{(X_{1, i} ,X_{2, j} ,Y) \in \mathcal{T}_{\varepsilon}^{(T)}\}.
$$
By symmetry, we assume, without loss of generality, that (1,1) was
sent. Define $\pi^{(T)}$ to be the event that
$$
\int_0^T (X_{1, 1}(t))^2 dt > P_1 T, \quad \int_0^T (X_{2, 1}(t))^2 dt > P_2 T.
$$
Then, $\hat{P}_e^{(T)}$, the error probability for the above coding scheme (where codewords violating the power constraint are allowed), can be upper bounded as follows:
$$
\hat{P}_e^{(T)} = P(\pi^{(T)} \cup E_{11}^c \bigcup \cup_{(i, j) \neq (1, 1)} E_{ij})
$$
$$
\leq P(\pi^{(T)})+P(E_{11}^c)+\sum_{i \neq 1, j =1} P(E_{i1})+\sum_{i=1, j \neq 1} P(E_{1j})+\sum_{i \neq 1, j \neq 1} P(E_{ij}).
$$
So, for any $i, j \neq 1$, we have
$$
\hat{P}^{(T)}_e \leq P(\pi^{(T)})+ P(E^c_{11})+e^{T R_1} P(E_{i1})+e^{T R_2} P(E_{1j})+ e^{T R_1 + T R_2} P(E_{ij})
$$
Using the well-known fact that an OU process is ergodic~\cite{le08}, we deduce that $P(\pi^{(T)}) \to 0$ as $T \to \infty$. And by Lemma~\ref{information-stability} and Theorem $6.2.1$ in~\cite{ih93}, we have
$$
\lim_{T \rightarrow \infty} P((X_{1, 1}, X_{2, 1}, Y) \in \mathcal{T}_{\varepsilon}^{(T)})=1 \mbox{ and thus } \lim_{T\rightarrow \infty} P(E^c_{11})=0.
$$
Now, we have for any $i \neq 1$,
\begin{align*}
P(E_{i1}) & = P((X_{1, i}, X_{2, 1}, Y) \in \mathcal{T}^{(T)}_{\varepsilon}) \\
&= \int_{(x_1, x_2, y) \in \mathcal{T}_{\varepsilon}^{(T)}} d\mu_{X_1}(x_1) d\mu_{X_2 Y}(x_2, y) \\
&= \int_{\mathcal{T}_{\varepsilon}^{(T)}}\frac{1}{\varphi_1(x_1, x_2, y)}d\mu_{X_1 X_2 Y}(x_1, x_2, y) \\
&\leq \int_{\mathcal{T}_{\varepsilon}^{(T)}}e^{-I_T(X_1; X_2, Y)+\varepsilon T}d\mu_{X_1 X_2 Y}(x_1, x_2, y) \\
&= e^{-I_T(X_1; Y|X_2)+\varepsilon T},
\end{align*}
where we have used the independence of $X_1$ and $X_2$, and the consequent fact that
$$
I_T(X_1; X_2, Y)=I_T(X_1; X_2)+I_T(X_1; Y| X_2)=I_T(X_1; Y|X_2).
$$
Similarly, we have, for $j \neq 1$,
$$
P(E_{1j}) \leq e^{-I_T(X_2; Y|X_1)+\varepsilon T},
$$
and for $i, j \neq 1$,
$$
P(E_{ij}) \leq e^{-I_T(X_1, X_2; Y)+\varepsilon T}.
$$
It then follows that
$$
\hat{P}^{(T)}_e \leq P(\pi^{(T)})+P(E_{11}^c)+ e^{T R_1+ \varepsilon T-I_T(X_1; Y|X_2)}+e^{T R_2+ \varepsilon T-I_T(X_2; Y|X_1)}+e^{T R_1+T R_2+ \varepsilon T-I_T(X_1, X_2; Y)}.
$$
By Lemma~\ref{independent-OUs}, one can choose independent OU processes $X_1, X_2$ such that
$I_T(X_1; Y|X_2)/T \rightarrow (P_1-\eps)/2$, $I_T(X_2; Y|X_1)/T \rightarrow (P_2-\eps)/2$ and $I_T(X_1, X_2; Y)/T \rightarrow (P_1+P_2-2 \eps)$ uniformly in $T$. This implies that with $\eps$ chosen sufficiently small, we have $\hat{P}^{(T)}_e \rightarrow 0$, as $T \rightarrow \infty$. In other words, there exists a sequence of good codes (which may not satisfy the power constraint) with low average error probability. Now, from each of the above codes, we delete the worse half of the codewords (any codeword violating the power constraint will be deleted since it must have error probability $1$). Then, with only slightly decreased transmission rate, the remaining codewords will satisfy the power constraint and will have small maximum error probability (and thus small average error probability $P_e^{(T)}$), which implies that the rate pair $(R_1,R_2)$ is achievable.
\end{proof}

\begin{rem}
The achievability part can be proven alternatively, which will be roughly described as follows: for arbitrarily small $\eps > 0$, by Lemma~\ref{independent-OUs}, one can choose independent OU processes $X_i$ with respective variances $P_i-\eps$, $i=1, 2$, such that $I_T(X_i; Y)/T$ approaches $(P_i-\eps)/2$. Then, a parallel random coding argument with $X_j$, $j \neq i$, being treated as noise at receiver $i$ shows that the rate pair $((P_1-\eps)/2, (P_2-\eps)/2)$ can be approached, which yields the achievability part.
\end{rem}

\section{Gaussian ICs} \label{IC}

Consider the following continuous-time white Gaussian interference channel having no feedback and with $m$ pairs of senders and receivers: for $i=1, 2, \ldots, m$,
\begin{align}  \label{Equation-IC}
Y_i(t) =a_{i1} \int_0^t X_1(s, M_1) ds + a_{i2} \int_0^t X_2(s, M_2) ds + \cdots + a_{im} \int_0^t X_m(s, M_m) ds+ B_i(t), \quad t \geq 0,
\end{align}
where $X_i$ is the continuous channel input from sender $i$, which depends on $M_i$, the message sent from sender $i$, which is independent of all messages from other senders, and $a_{ij} \in \mathbb{R}$, $i, j=1, 2, \ldots, m$, is the channel gain from sender $j$ to receiver $i$, all $B_i(t)$ are (possibly correlated) standard Brownian motions.

For $T, R_1, \ldots, R_m, P_1, \ldots, P_m > 0$, a $(T, (e^{T R_1}, \ldots, e^{T R_m}), (P_1, \ldots, P_m))$-code for the IC (\ref{Equation-IC}) consists of $m$ sets of integers $\mathcal{M}_i=\{1, 2, \ldots, e^{T R_i}\}$, the {\it message alphabet} for user $i$, $i=1, 2, \ldots, m$, and $m$ {\it encoding functions}, $X_i: \mathcal{M}_i \rightarrow C[0, T]$ satisfying the following power constraint: for any $i=1, 2, \ldots, m$, 
\begin{equation} \label{PowerConstraint-IC}
\frac{1}{T} \int_0^T X_i^2(s, M_i) ds \leq P_i,
\end{equation} 
and $m$ {\it decoding functions}, $g_i: C[0, T] \rightarrow \mathcal{M}_i$, $i=1, 2, \ldots, m$.

The average probability of error for the $(T, (e^{T R_1}, \ldots, e^{T R_m}), (P_1, \ldots, P_m))$-code is defined as
$$
\hspace{-1cm} P_e^{(T)}=\frac{1}{e^{T(\sum_{i=1}^m R_i)}} \sum_{(M_1, M_2, \ldots, M_m) \in \mathcal{M}_1 \times \mathcal{M}_2 \times \cdots \times \mathcal{M}_m} P\{g_i(Y_{i, 0}^T) \neq M_i, i=1, 2, \ldots, m~|~(M_1, M_2, \ldots, M_m) \mbox{ sent}\}.
$$
A rate tuple $(R_1, R_2, \ldots, R_m)$ is said to be {\bf achievable} for the IC if there exists a sequence of $(T, (e^{T R_1}, \ldots, e^{T R_m}),(P_1, \ldots, P_m))$-codes with $P_e^{(T)} \rightarrow 0$ as $T \rightarrow \infty$. The {\bf capacity region} of the IC is the closure of the set of all the achievable $(R_1, R_2, \ldots, R_m)$ rate tuples.

The following theorem explicitly characterizes the capacity region of the above IC:
\begin{thm} \label{Theorem-IC}
The capacity region of the continuous-time white Gaussian IC (\ref{Equation-IC}) is
$$
\{(R_1, R_2, \ldots, R_m) \in \mathbb{R}_+^m: R_i \leq a_{ii}^2 P_i/2, \quad i=1, 2, \ldots, m\}.
$$
\end{thm}

\begin{proof}
For notational convenience only, we only prove the case when $n=2$; the case when $n$ is generic being similar.

{\bf The converse part.} In this part, we will show that for any sequence of $(T, (e^{T R_1}, e^{T R_2}), (P_1, P_2))$ codes with
$P_e^{(T)} \rightarrow 0$, the rate pair $(R_1, R_2)$ will have to satisfy
\begin{equation} \label{less-than-squared}
R_1 \leq a_{11}^2 P_1/2, \quad R_2 \leq a_{22}^2 P_2/2.
\end{equation}

Fix $T$ and consider the above-mentioned $(T, (e^{T R_1}, e^{T R_2}), (P_1, P_2))$ code. By the code construction, for $i=1, 2$, it is possible to estimate the messages $M_i$ from the channel output $Y_{i, 0}^T$ with an arbitrarily low probability of error. Hence, by Fano's inequality, for $i=1, 2$,
$$
H(M_i|Y_{i, 0}^T) = T \varepsilon_{i, T},
$$
where $\varepsilon_{i, T} \to 0$ as $T \to \infty$. We then have
$$
T R_1=H(M_1)=H(M_1|M_2)=I(M_1; Y_1|M_2)+H(M_1|M_2, Y_1) \leq I(M_1; Y_1|M_2)+ T \varepsilon_{1, T},
$$
As in the proof of Theorem~\ref{Theorem-MAC}, we have
$$
I(M_1; Y_{1, 0}^T|M_2)= \frac{a^2_{11}}{2} \int_0^T E[(X_1(s)-E[X_1(s)|M_2, Y_{1, 0}^s])^2] ds.
$$
It then follows that
$$
T R_1 \leq \frac{a_{11}^2}{2} \int_0^T E[(X_1(s)-E[X_1(s)|M_2, Y_{1, 0}^s])^2] ds + T \varepsilon_{1, T},
$$
which implies that $R_1 \leq a_{11}^2 P_1/2$. With a parallel argument, one can derive that $R_2 \leq a_{22}^2 P_2/2$. The proof for the converse part is then complete.

{\bf The achievability part.} We only sketch the proof of this part. For arbitrarily small $\eps > 0$, by Lemma~\ref{independent-OUs}, one can choose independent OU processes $X_i$ with respective variances $P_i-\eps$, $i=1, 2$, such that $I_T(X_i; Y)/T$ approaches $a^2_{ii}(P_i-\eps)/2$. Then, a parallel random coding argument as in the proof of Theorem~\ref{Theorem-MAC} with $X_j$, $j \neq i$, being treated as noise at receiver $i$ shows that the rate pair $(a^2_{11}(P_1-\eps)/2, a^2_{22}(P_2-\eps)/2)$ can be approached, which yields the achievability part.
\end{proof}

\section{Gaussian BCs} \label{BC}

In this section, we consider a continuous-time white Gaussian BC with $m$ receivers, which is characterized by: for $i=1, 2, \ldots, m$,
\begin{equation} \label{Equation-BC}
Y_i(t) = \sqrt{snr_i} \int_0^t X(s, M_1, M_2, \ldots, M_m) ds + B_i(t), \quad t \geq 0,
\end{equation}
where $X$ is the continuous channel input, which depends on $M_i$, the message sent from sender $i$, which is uniformly distributed over a finite alphabet $\mathcal{M}_i$ and independent of all messages from other senders, $snr_i$ is the signal-to-noise ratio (SNR) in the channel for user $i$, $B_i(t)$ are (possibly correlated) standard Brownian motions.

For $T, R_1, R_2, \ldots, R_m, P> 0$, a $(T, (e^{T R_1}, \ldots, e^{T R_m}), P)$-code for the BC (\ref{Equation-BC}) consists of $m$ set of integers $\mathcal{M}_i=\{1, 2, \ldots, e^{T R_i}\}$, the {\it message set} for receiver $i$, $i=1, 2, \ldots, m$,  and an {\it encoding function}, $X: \mathcal{M}_1 \times \mathcal{M}_2 \times \cdots \times \mathcal{M}_m \rightarrow C[0, T]$, which satisfies the following power constraint: 
\begin{equation} \label{PowerConstraint-BC}
\frac{1}{T} \int_0^T X^2(s, M_1, M_2, \ldots, M_m) ds \leq P,
\end{equation}
and $m$ {\it decoding functions}, $g_i: C[0, T] \rightarrow \mathcal{M}_i$, $i=1, 2, \ldots, m$.

The average probability of error for the $(T, (e^{T R_1}, e^{T R_2}, \ldots, e^{T R_m}), P)$-code is defined as
$$
\hspace{-1cm} P_e^{(T)}=\frac{1}{e^{T(\sum_{i=1}^m R_i)}} \sum_{(M_1, M_2, \ldots, M_m) \in \mathcal{M}_1 \times \mathcal{M}_2 \times \cdots \times \mathcal{M}_m} P\{g_i(Y_0^T) \neq M_i, i=1, 2, \ldots, m~|~(M_1, M_2, \ldots, M_m) \mbox{ sent}\}.
$$
A rate tuple $(R_1, R_2, \ldots, R_m)$ is said to be {\bf achievable} for the BC if there exists a sequence of $(T, (e^{T R_1}, e^{T R_2}, \ldots, e^{T R_m}), P)$-codes with $P_e^{(T)} \rightarrow 0$ as $T \rightarrow \infty$. The {\bf capacity region} of the BC is the closure of the set of all the achievable $(R_1, R_2, \ldots, R_m)$ rate tuples.

The following theorem explicitly characterizes the capacity region of the above BC:
\begin{thm} \label{Theorem-BC}
The capacity region of the continuous-time white Gaussian BC (\ref{Equation-BC}) is
$$
\left\{(R_1, R_2, \ldots, R_m) \in \mathbb{R}_+^m: \frac{R_1}{snr_1}+\frac{R_2}{snr_2}+\cdots+\frac{R_m}{snr_m} \leq \frac{P}{2}\right\}.
$$
\end{thm}

To prove Theorem~\ref{Theorem-BC}, we will first need the following lemma.
\begin{lem} \label{Xianming-Lemma}
Consider a continuous-time white Gaussian channel characterized by the following equation
$$
Y(t)=\sqrt{snr} \int_0^t X(s) ds + B(t), \quad t \geq 0,
$$
where $snr \geq 0$ denotes the SNR in the channel and $M$ is the message to be transmitted through the channel. Then, for any fixed $T$, $I_T(M; Y)/snr$ is a monotone decreasing function of $snr$.
\end{lem}

\begin{proof}
For notational convenience, in this proof, we write $I_T(M; Y)$ as $I_T(snr)$.
By Theorem $6.2.1$ in~\cite{ih93}, we have
$$
I_T(snr)=\frac{snr}{2} \int_0^T E[(X(s)-E[X(s)|Y_0^s])^2] ds,
$$
and Theorem $6$ in~\cite{gu05}, we have (the derivative is with respect to $snr$)
$$
I'_T(snr)=\frac{1}{2} \int_0^T E[(X(s)-E[X(s)|Y_0^T])^2] ds.
$$
It then follows that
\begin{align*}
\left(\frac{I_T(snr)}{snr}\right)' &=\frac{1}{snr}\left(I'_T(snr)-\frac{I_T(snr)}{snr}\right) \\
&=\frac{1}{2 snr} \left(\int_0^T E[(X(s)-E[X(s)|Y_0^T])^2] ds-\int_0^T E[(X(s)-E[X(s)|Y_0^s])^2] ds \right) \leq 0,
\end{align*}
which immediately implies the lemma.
\end{proof}

\begin{proof}[Proof of Theorem~\ref{Theorem-BC}]
For notational convenience only, we prove the case when $n=2$, the case when $n$ is generic being parallel.

{\bf The converse part.} Without loss of generality, we assume that
$$
snr_1 \geq snr_2.
$$
We will show that for any sequence of $(T, (e^{T R_1}, e^{T R_2}), P)$ codes with
$P_e^{(T)} \rightarrow 0$ as $T \rightarrow \infty$, the rate pair $(R_1, R_2)$ will have to satisfy
\begin{equation} \label{divided-by-snr}
\frac{R_1}{snr_1}+\frac{R_2}{snr_2} \leq \frac{P}{2}.
\end{equation}

Fix $T$ and consider the above-mentioned $(T, (e^{T R_1}, e^{T R_2}), P)$-code. By the code construction, for $i=1, 2$, it is possible to estimate the messages $M_i$ from the channel output $Y_{i, 0}^T$ with an arbitrarily low probability of error. Hence, by Fano's inequality, for $i=1, 2$,
$$
H(M_i|Y_{i, 0}^T) \leq T R_i P^{(T)}_e +H(P^{(T)}_e) = T \varepsilon_{i, T},
$$
where $\varepsilon_{i, T} \rightarrow 0$ as $T \rightarrow \infty$. It then follows that
\begin{equation} \label{eq-1}
T R_1 = H(M_1) = H(M_1|M_2) \leq I(M_1; Y_{1, 0}^T|M_2) + T \varepsilon_{1, T},
\end{equation}
\begin{equation} \label{eq-2}
T R_2 = H(M_2) \leq I(M_2; Y_{2, 0}^T) + T \varepsilon_{2, T}.
\end{equation}
By the chain rule of mutual information, we have
\begin{equation} \label{eq-3}
I(M_1, M_2; Y_{2, 0}^T)=I(M_2; Y_{2, 0}^T)+I(M_1; Y_{2, 0}^T|M_2) \geq I(M_2; Y_{2, 0}^T) + \frac{snr_2}{snr_1} I(M_1; Y_{1, 0}^T|M_2),
\end{equation}
where, for the inequality above, we have applied Lemma~\ref{Xianming-Lemma}.
Now, by Theorem $6.2.1$ in~\cite{ih93}, we have
$$
I(M_1, M_2; Y_{2, 0}^T) = \frac{snr_2}{2} \int_0^T E[(X(s)-E[X(s)|Y_{2, 0}^s])^2] ds \leq \frac{snr_2}{2} \int_0^T E[X^2(s)] ds,
$$
which, together with (\ref{eq-1}), (\ref{eq-2}), (\ref{eq-3}) and (\ref{PowerConstraint-BC}), immediately implies the converse part.

{\bf The achievability part.} We only sketch the proof of this part. For an arbitrarily small $\eps > 0$, by Theorem $6.4.1$ in~\cite{ih93}, one can choose an OU processes $\tilde{X}$ with variance $P-\eps$, such that $I_T(\tilde{X}; Y_i)/T$ approaches $snr_i (P-\eps)/2$. For any $0 \leq \lambda \leq 1$, let
$$
X(t)=\sqrt{\lambda} X_1(t)+ \sqrt{1-\lambda} X_2(t), \quad t \geq 0,
$$
where $X_1$ and $X_2$ are independent copies of $\tilde{X}$. Then, by a similar argument as in the proof of Lemm~\ref{independent-OUs}, we deduce that $I_T(X_1; Y_1)/T, I_T(X_2; Y_2)/T$ approach $snr_1 \lambda (P-\eps)/2$, $snr_2 (1-\lambda)(P-\eps)/2$, respectively. Then, a parallel random coding argument as in the proof of Theorem~\ref{Theorem-MAC} such that
\begin{itemize}
\item when encoding, $X_i$ only carries the message meant for receiver $i$;
\item when decoding, receiver $i$ treats $X_j$, $j \neq i$, as noise,
\end{itemize}
shows that the rate pair $(snr_1 \lambda (P-\eps)/2, snr_2 (1-\lambda) (P-\eps)/2)$ can be approached, which immediately establishes the achievability part.
\end{proof}

\begin{rem}
For the achievability part, instead of using the power sharing scheme as in the proof, one can also employ the following time sharing scheme: set $X$ to be $X_1$ for $\lambda$ fraction of the time, and $X_2$ for $1-\lambda$ fraction of the time. Then, it is straightforward to check this scheme also achieves the rate pair $(snr_1 \lambda (P-\eps)/2, snr_2 (1-\lambda) (P-\eps)/2)$. This, from a different perspective, echoes the observation in~\cite{la03} that time sharing achieves the capacity region of a white Gaussian BC as the bandwidth limit tends to infinity. 
\end{rem}

\section{Coding Theorems for Repeated Channels} \label{Repeated-Channel}

In this section, we consider the so-called {\em repeated versions}~\cite{ih94} of (\ref{Equation-MAC}), (\ref{Equation-IC}) and (\ref{Equation-BC}) and we prove coding theorems as the number of repeated times tends to infinity. 

We start off with the repeated version of a continuous-time white Gaussian MAC. For a fixed $T_0 > 0$, consider a continuous-time white Gaussian MAC with $m$ users and possible feedback:
\begin{equation}  \label{Equation-Repeated-MAC}
\hspace{-0.5cm} Y(t)=\int_0^t X_1(s, M_1, Y_0^s) ds+\int_0^t X_2(s, M_2, Y_0^s) ds + \cdots +\int_0^t X_m(s, M_m, Y_0^s) ds+B(t), \quad 0 \leq t \leq N T_0,
\end{equation}
where $X_i$ is the continuous channel input from sender $i$, which depends on $M_i$, the message sent from sender $i$, which is independent of all messages from other senders, and possibly on $Y_0^{s}$, the channel output up to time $s$. 

For $N, R_1, R_2, \ldots, R_m > 0$, a $(N, (e^{N R_1}, e^{N R_2}, \ldots, e^{N R_m}), (P_1, \ldots, P_m))$-code for the MAC (\ref{Equation-Repeated-MAC}) consists of a set $S$ of $l$ signal waveforms with $S=\{s_i \in C[0, T_0]: i=1, 2, \ldots, l\}$, $m$ sets of integers $\mathcal{M}_i=\{1, 2, \ldots, e^{N R_i}\}$, the {\it message alphabet} for user $i$, $i=1, 2, \ldots, m$, and $m$ {\it encoding functions}, $X_i: \mathcal{M}_i \rightarrow C[0, N T_0]$ such that for any $i$, $X_{i, j T_0}^{(j+1) T_0} \in S$ (in other words, the waveform $X_i$, when restricted on the interval $[j T_0, (j+1) T_0]$), is from $S$), for all $j=0, 2, \ldots, N-1$ and satisfy the following power constraint: for any $i$, 
\begin{equation}  \label{PowerConstraint-Repeated-MAC}
\frac{1}{N T_0} \int_0^{N T_0} X^2_i(s, M_i, Y_0^s) ds \leq P_i, \quad i = 1, 2, \ldots, m,
\end{equation}
and a {\it decoding function},
$$
g: C[0, N T_0] \rightarrow \mathcal{M}_1 \times \mathcal{M}_2 \times \cdots \times \mathcal{M}_m.
$$
The average probability of error probability for the $(N, (e^{N R_1}, \ldots, e^{N R_m}), (P_1, \ldots, P_m))$-code is defined as
$$
\hspace{-1cm} P_e^{(N)}=\frac{1}{e^{N (\sum_{i=1}^m R_i)}} \sum_{(M_1, M_2, \ldots, M_m) \in \mathcal{M}_1 \times \mathcal{M}_2 \times \cdots \times \mathcal{M}_m} P\{g(Y_0^{N T_0}) \neq (M_1, M_2, \ldots, M_m)~|~(M_1, M_2, \ldots, M_m) \mbox{ sent}\}.
$$
A rate tuple $(R_1, R_2, \ldots, R_m)$ is said to be {\bf achievable} for the MAC if there exists a sequence of $(N, (e^{N R_1}, \ldots, e^{N R_m}), (P_1, \ldots, P_m))$-codes with $P_e^{(N)} \rightarrow 0$ as $N \rightarrow \infty$. The {\bf capacity region} of the MAC is the closure of the set of all the achievable $(R_1, R_2, \ldots, R_m)$ rate tuples.

The following theorem gives an explicit characterization of the capacity region of the above-mentioned repeated channel. The proof of the proof follows a largely parallel argument as in Theorem~\ref{Theorem-MAC}, and thus omitted.
\begin{thm}  \label{Theorem-Repeated-MAC}
Whether there is feedback or not, the capacity region of the above-mentioned continuous-time repeated white Gaussian MAC is
$$
\{(R_1, R_2, \ldots, R_m) \in \mathbb{R}_+^m: R_i \leq P_i T_0/2, \quad i=1, 2, \ldots, m\}.
$$
\end{thm}

Similarly as above, (\ref{Equation-IC}) and (\ref{Equation-BC}) have repeated versions and their capacity regions can be defined in a similar fashion as above.

The following theorem explicitly characterizes the capacity region of a repeated IC:
\begin{thm} \label{Theorem-Repeated-IC}
The capacity region of the repeated version of the continuous-time white Gaussian IC (\ref{Equation-IC}) is
$$
\{(R_1, R_2, \ldots, R_m) \in \mathbb{R}_+^m: R_i \leq a_{ii}^2 P_i T_0/2, \quad i=1, 2, \ldots, m\}.
$$
\end{thm}

The following theorem explicitly characterizes the capacity region of a repeated BC:
\begin{thm} \label{Theorem-Repeated-BC}
The capacity region of the repeated version of the continuous-time white Gaussian BC (\ref{Equation-BC}) is
$$
\left\{(R_1, R_2, \ldots, R_m) \in \mathbb{R}_+^m: \frac{R_1}{snr_1}+\frac{R_2}{snr_2}+\cdots+\frac{R_m}{snr_m} \leq \frac{P T_0}{2}\right\}.
$$
\end{thm}

\section{Conclusions and Future Work}

In this paper, we have proposed to use Brownian motions (instead of white Gaussian noises) to formulate some continuous-time multi-user networks. Such a formulation allows us to carry over the established techniques and tools from the discrete-time setting over to the continuos-time one, and thereby derive explicit characterizations of the infinite bandwidth capacity regions of a continuous-time white Gaussian multiple access channel with/without feedback, a continuous-time white Gaussian interference channel without feedback and a continuous-time white Gaussian broadcast channel without feedback.

While it is proven that for white Gaussian MACs, feedback does not increase the infinite bandwidth capacity region, and it remains to be seen whether feedback will increase the infinite bandwidth capacity regions for white Gaussian ICs, BCs or other channels. 

Note that there exist in-depth studies~\cite{hi74, hi75, ih80, ih90} on continuous-time point-to-point colored Gaussian channels with possible feedback, natural generalizations for white Gaussian channels. In this regard, another possible direction is to look into whether the ideas and techniques in this paper can be applied to more general networks, such as multi-hop channels with more general noises.

The sampling theorems in Section~\ref{time-sampling} have established a continuous-time Gaussian channel as the limit of a sequence of discrete-time ones, in an information-theoretic sense. Another interesting direction is to further quantify the connections between the continuous-time channels and their discrete-time counterparts. For instance, one can examine, perhaps with strengthened conditions, whether Theorems~\ref{sampling-output} and~\ref{sampling-input-output} still hold true for any sequence of $\Delta_n$ with shrinking sampling intervals, that is, $\delta_n(\Delta_n)$ tends to $0$ as $n$ tends to infinity, where for a given $\Delta_n$, $\delta(\Delta_n)\triangleq \min \{t_{n, i}-t_{n, i-1}: i=1, 2, \ldots, n\}$. Also, a more quantitive analysis on how fast these discrete-time channels will ``approach'' the continuous-time one would significantly enhance our understanding of both types of channels. As a byproduct, such a connection may also provide us an alternative way to derive/estimate 
the capacity regions of discrete-time multi-user Gaussian channels, which have largely remained unknown. 
 
\bigskip

{\bf Acknowledgement.} We would like to thank Jun Chen, Young-Han Kim, Tsachy Weissman, Wenyi Zhang for insightful suggestions and comments, and for pointing out relevant references.

\end{document}